\documentclass{article}

\usepackage[top=1in, bottom=1in, left=1in, right=1in]{geometry}

\usepackage{amsmath,amssymb}
\usepackage{mathtools}

\usepackage{thmtools}

\usepackage{graphicx}
\usepackage{float}
\usepackage{booktabs}
\usepackage{subcaption}
\usepackage{tikz-cd}

\usepackage{natbib}
\usepackage{url}
\usepackage{hyperref}

\usepackage{authblk}

\usepackage{xcolor}
\usepackage[colorinlistoftodos]{todonotes}
\usepackage{comment}
\usepackage{verbatim}

\usepackage{amsthm}
\newtheorem{theorem}{Theorem}
\newtheorem{proposition}[theorem]{Proposition}
\newtheorem{lemma}[theorem]{Lemma}
\newtheorem{corollary}[theorem]{Corollary}

\theoremstyle{definition}
\newtheorem{definition}[theorem]{Definition}

\theoremstyle{remark}
\newtheorem{remark}[theorem]{Remark}

\theoremstyle{plain}
\newtheorem*{maintheorem}{Main Theorem}

\title{Linear Self-Symmetries of Inverse-Based S-boxes: Generic Rigidity and Applications}

\author{
Zheng Zhang \thanks{Corresponding author. Department of Mathematics, Towson University, 7800 York Rd, Towson, MD 21204, USA. Email: \texttt{zhengzhang@towson.edu}.}
Na Zhang \thanks{Department of Mathematics, Towson University, 7800 York Rd, Towson, MD 21204, USA. Email: \texttt{nzhang@towson.edu}.}
}

\date{}

\begin{document}

\maketitle
\begin{abstract}
Finite-field inversion is a basic nonlinear building block in the
design of S-boxes for symmetric cryptography. In inverse-based S-boxes, finite-field inversion is often combined with
linear or affine layers that may preserve or destroy algebraic symmetries of
the inverse function. We study how often nontrivial linear
self-symmetries remain as the outer invertible linear transformation
varies over \(GL_n(\mathbb F_2)\). Using the classification of
two-sided linear equivalences of inversion, we reduce the problem to
semilinear conjugacy and centralizer estimates. We obtain matching
upper and lower bounds for the probability that a nontrivial linear
self-symmetry remains, and these bounds determine its leading
quadratic exponent. This exponent is governed by the smallest prime divisor of the extension
degree. Thus the decay rate depends on the arithmetic of the field
extension and not only on its dimension. As cryptographic applications, we provide a practical rigidity
criterion, prove that the full AES S-box has no nontrivial linear
self-symmetry, and illustrate the broader relevance of the framework
through ARIA and Vision.
\end{abstract}


\section{Introduction}
Finite-field inversion on $\mathbb{F}_{2^n}$ is one of the best-known
nonlinear maps used in the design of S-boxes for symmetric-key
cryptography. Its algebraic structure is also highly symmetric. If
\[
J(x)=x^{-1},\qquad J(0)=0,
\]
then \(J\) commutes with every Frobenius automorphism
\(x\mapsto x^{2^r}\). Thus inversion comes with a natural family of
nontrivial linear symmetries.

In inverse-based S-box constructions, inversion is usually
combined with an outer linear or affine layer, and such a layer may
preserve or destroy symmetries of the inverse function. This raises a natural
structural question: when the outer linear layer varies, how often does
any nontrivial symmetry of inversion survive? 

To make the question precise, for a map
\(\Phi:\mathbb F_{2^n}\to\mathbb F_{2^n}\), define
\[
\operatorname{Stab}(\Phi)
=
\{T\in GL_n(\mathbb F_2):T\circ\Phi=\Phi\circ T\}.
\]
A nonidentity element of \(\operatorname{Stab}(\Phi)\) is a nontrivial
linear self-symmetry of \(\Phi\), and we call \(\Phi\) rigid when
\(\operatorname{Stab}(\Phi)=\{I\}\).
\begin{remark}
We emphasize that this notion is different from linear
self-equivalence. A linear self-equivalence allows a pair of linear
maps acting separately on the input and output, whereas a linear
self-symmetry in the sense above requires the same linear map to
commute with \(\Phi\).
\end{remark}

We study the family
\[
\Phi_A=A\circ J,
\qquad A\in GL_n(\mathbb F_2),
\]
and measure the frequency with which \(\Phi_A\) retains a nontrivial
linear self-symmetry.

Self-symmetries are especially natural in iterative
constructions. When the same linear transformation is compatible with
the S-box, the diffusion layer, and the round constants, it can
propagate from one round to the next without changing form. Such
self-symmetries therefore represent algebraic structure that can
persist through repeated applications of the round function and may
give rise to invariant state spaces or other non-random structure~\cite{leander2011invariant}.
Understanding when such self-symmetries survive is therefore relevant
to both the design and the structural analysis of inverse-based
S-boxes.

This setting arises naturally in several inverse-based cryptographic
constructions. The AES S-box is the best-known example, where inversion
over \(\mathbb F_{2^8}\) is followed by a fixed affine layer
\cite{fips197}. Related inverse-based substitution layers also appear
in ARIA~\cite{lee2010aria}, while more recent designs such as Vision
use inversion together with sparse linearized polynomials over binary
fields~\cite{aly2020marvellous}. These examples show that the interaction
between finite-field inversion and surrounding linear structure is a
recurring feature of inverse-based cryptographic design.

This problem is related to several lines of work on symmetry,
equivalence, and self-equivalence of cryptographic functions.
Rijmen, Barreto, and Gazzoni Filho~\cite{rijmen2008rotation} studied
rotation symmetry in S-boxes generated algebraically from finite-field
power maps. Related work has further examined rotation-symmetric
S-boxes and their affine-equivalence classes
\cite{kavut2012rotation,omland2022rotation}. Algebraic and alternative
representations of Rijndael have also been studied extensively
\cite{rosenthal2003polynomial,liu2005rijndael,barkan2002rijndael}.
These works show that finite-field structure and surrounding linear
transformations can produce useful symmetries and alternative
representations of S-boxes, but they address a different question from
the survival problem considered here.

More generally, linear, affine, EA, and CCZ equivalence provide
standard frameworks for comparing vectorial Boolean functions
\cite{carlet1998codes,biryukov2003toolbox,canteaut2019ccz,chen2021affine}.
Self-equivalences have also been studied in white-box cryptography
\cite{ranea2021self}. For inversion specifically,
K{\"o}lsch~\cite{kolsch2021ccz} obtained strong classifications of its
CCZ and EA equivalence classes. These notions allow broader input and output transformations than the
commuting linear self-symmetries considered here.

More directly related to our problem, Yuan et
al.~\cite{yuan2019inverse} classified the two-sided linear
equivalences
\[
J\circ T=U\circ J.
\]
Their classification identifies the possible transformations that can
occur in the stabilizer of \(A\circ J\). Ishizuka~\cite{ishizuka2026stabilizers}
recently proved generic stabilizer triviality in the full space of
vectorial functions. In contrast, we consider the structured family
\[
\{A\circ J:A\in GL_n(\mathbb F_2)\}.
\]
The above works do not determine how frequently a nontrivial linear
self-symmetry survives in this family, or how that frequency depends
on the arithmetic of \(n\).
Let
\[
\mathcal E_n
=
\{A\in GL_n(\mathbb F_2):
\operatorname{Stab}(A\circ J)\ne\{I\}\}
\]
and
\[
p_n=\frac{|\mathcal E_n|}{|GL_n(\mathbb F_2)|}.
\]
Thus \(p_n\) is the probability that a uniformly chosen outer linear
transformation leaves a nontrivial linear self-symmetry of the
transformed inversion map. Our main result determines the exact leading
quadratic exponent of \(p_n\). 
\begin{maintheorem}
\label{thm:intro-main}
Let \(q(n)\) be the smallest prime divisor of \(n\), and let
\[
\kappa=\prod_{j\ge1}(1-2^{-j}).
\]
For \(n\ge4\),
\[
\kappa^{q(n)}2^{-n^2(1-1/q(n))}
\le p_n\le
\frac n\kappa 2^{-n^2(1-1/q(n))+n}.
\]
Consequently,
\[
-\log_2 p_n
=
n^2\left(1-\frac1{q(n)}\right)+O(n).
\]
\end{maintheorem}
In particular, when \(n\) is even the leading quadratic term is
\(n^2/2\), while for prime \(n\) it is \(n^2\). Thus the rate at which
nontrivial self-symmetries disappear depends strongly on the arithmetic
of the extension degree. 

The theorem can be read in two ways. For designers who want symmetry, \(\mathcal E_n\) measures the set of
outer linear layers that preserve a nontrivial self-symmetry of
inversion. For designers who want to avoid symmetry, the theorem shows
that a generic outer layer removes all such symmetries. Common S-box
dimensions such as \(n=4\) and \(n=8\) lie in the least favorable case
\(q(n)=2\), which motivates direct checks for concrete constructions.

We develop several cryptographic applications in Section~5. We provide a simple rigidity criterion and apply it to the AES S-box, use ARIA to
illustrate the propagation of byte-wise self-symmetries through binary
diffusion, and consider Vision as a large-field example. We also briefly
discuss Camellia~\cite{rfc3713} and SM4~\cite{SM4ISO} as related inverse-based constructions.

The paper is organized as follows. Section~2 introduces the notation
and the basic symmetry structure of finite-field inversion. Section~3
develops the structural estimates for the exceptional set. Section~4
proves the main theorem. Section~5 develops cryptographic applications
to inverse-based S-boxes, including AES, ARIA, and Vision. Section~6
concludes the paper.

\section{Notation and Inversion Symmetries} 
In this section, we fix the notation and recall the basic structural facts about finite-field inversion and its linear symmetries that will be used throughout the paper. Let \(J\) denote the inverse function on \(\mathbb{F}_{2^n}\), extended by \(J(0)=0\). For \(A\in GL_n(\mathbb{F}_2)\), we write
\[
\Phi_A = A\circ J.
\]
This family is natural in the study of inversion-based S-boxes, where
an outer linear layer may preserve or destroy linear self-symmetries
of inversion.

\begin{definition}
For a map
\(\Phi:\mathbb F_{2^n}\to\mathbb F_{2^n}\), its linear stabilizer is
\[
\operatorname{Stab}(\Phi)
=
\{T\in GL_n(\mathbb F_2):T\circ\Phi=\Phi\circ T\}.
\]
For the family \(\Phi_A=A\circ J\), define the set
\[
\mathcal E_n
=
\left\{
A\in GL_n(\mathbb F_2):
\operatorname{Stab}(\Phi_A)\neq\{I\}
\right\},
\]
and let
\[
p_n
=
\frac{|\mathcal E_n|}{|GL_n(\mathbb F_2)|}.
\]
\end{definition}

A nonidentity element of \(\operatorname{Stab}(\Phi)\) represents a nontrivial
linear self-symmetry of \(\Phi\). Thus, triviality of the stabilizer means that
no nontrivial linear change of coordinates preserves the map. In particular,
\(p_n\) is the probability that a uniformly chosen outer linear transformation
\(A\) leaves a nontrivial linear self-symmetry in \(\Phi_A=A\circ J\). 

The central question is therefore which linear self-symmetries of the inverse
function can survive after composition with an outer linear transformation.
To address this question, we study the symmetries of \(J\),
which arise from the Frobenius action on \(\mathbb F_{2^n}\).

Let
$
\sigma(x)=x^2
$
denote the Frobenius automorphism of \(\mathbb F_{2^n}\).
Since inversion commutes with Frobenius,
$
J\circ\sigma=\sigma\circ J,
$
and hence
$
\langle \sigma\rangle\subseteq \operatorname{Stab}(J).$ Thus, the symmetry problem for \(\Phi_A\) begins with the interaction between the outer linear map \(A\) and the Frobenius symmetries inherent in \(J\). To make this interaction precise, we use the classification of two-sided linear equivalences of the inverse function due to Yuan et al~\cite{yuan2019inverse}. For \(a\in \mathbb F_{2^n}^{\times}\) and \(0\le r<n\), let
\[
M_a(x)=ax,
\qquad
T_{a,r}=M_a\circ \sigma^r,
\qquad
U_{a,r}=M_{a^{-1}}\circ \sigma^r.
\]
These are semilinear maps obtained by combining field multiplication with a Frobenius power. The following theorem gives the complete list of two-sided linear equivalences of \(J\) stated in the notation used here.

\begin{theorem}[Yuan et al.~\cite{yuan2019inverse}]
\label{thm:yuan}
Let \(n\ge 4\) and let
\(T,U\in GL_n(\mathbb F_2)\).
Then
\[
J\circ T=U\circ J
\]
if and only if there exist
\(a\in\mathbb F_{2^n}^{\times}\) and \(0\le r<n\) such that
\[
T=M_a\circ\sigma^r,
\qquad
U=M_{a^{-1}}\circ\sigma^r.
\]
\end{theorem}

The relevance of Theorem~\ref{thm:yuan} to our setting is immediate.
If \(T\in \operatorname{Stab}(\Phi_A)\), then
\[
T\circ A\circ J=A\circ J\circ T,
\]
or equivalently,
\[
J\circ T=(A^{-1}TA)\circ J.
\]
Thus Theorem~\ref{thm:yuan} shows that every linear self-symmetry of
\(\Phi_A\) must arise from one of the semilinear pairs
\[
T_{a,r}=M_a\circ\sigma^r,
\qquad
U_{a,r}=M_{a^{-1}}\circ\sigma^r,
\]
with
\[
A^{-1}T_{a,r}A=U_{a,r}.
\]
Consequently, the problem of determining whether \(\Phi_A\) has a
nontrivial linear stabilizer is reduced to a conjugacy problem between
these semilinear maps. To encode this condition, for
\(T,U\in GL_n(\mathbb F_2)\), let
\[
C(T,U)
=
\{A\in GL_n(\mathbb F_2):A^{-1}TA=U\}.
\]
Thus \(C(T,U)\) consists precisely of the outer linear transformations
that conjugate \(T\) to \(U\). The reduction above can now be stated
explicitly as follows.

\begin{proposition}
\label{prop:exceptional-union}
For \(n\ge 4\),
\[
\mathcal E_n
=
\bigcup_{(a,r)\ne(1,0)}
C(T_{a,r},U_{a,r}),
\]
where \(a\in\mathbb F_{2^n}^{\times}\) and \(0\le r<n\).
\end{proposition}

\begin{proof}
Suppose first that \(A\in\mathcal E_n\). Then
\(\operatorname{Stab}(\Phi_A)\) contains some nonidentity
\(T\in GL_n(\mathbb F_2)\). Since
\[
T\circ A\circ J=A\circ J\circ T,
\]
we have
\[
J\circ T=(A^{-1}TA)\circ J.
\]
By Theorem~\ref{thm:yuan}, there exist
\(a\in\mathbb F_{2^n}^{\times}\) and \(0\le r<n\) such that
\[
T=T_{a,r},
\qquad
A^{-1}T_{a,r}A=U_{a,r}.
\]
Because \(T\ne I\), we have \((a,r)\ne(1,0)\). Hence
\[
A\in C(T_{a,r},U_{a,r}).
\]

Conversely, suppose that
\[
A\in C(T_{a,r},U_{a,r})
\]
for some \((a,r)\ne(1,0)\). Then
\[
A^{-1}T_{a,r}A=U_{a,r},
\]
and, by Theorem~\ref{thm:yuan},
\[
J\circ T_{a,r}=U_{a,r}\circ J.
\]
Combining these identities gives
\[
T_{a,r}\circ A\circ J
=
A\circ J\circ T_{a,r}.
\]
Thus
\[
T_{a,r}\in \operatorname{Stab}(\Phi_A).
\]
Since \((a,r)\ne(1,0)\), we have \(T_{a,r}\ne I\), so
\(A\in\mathcal E_n\).
\end{proof}

Hence, estimating \(p_n\) reduces to controlling the sizes of these
conjugator sets. The structural results needed for this counting
problem are developed in the next section.

\section{Structural Lemmas for the Exceptional Set} 
Proposition~\ref{prop:exceptional-union} expresses the exceptional set as a union of semilinear conjugator sets. We now develop the structural lemmas needed to control the size of these sets and to identify the largest contributions to \(\mathcal E_n\).

For \(T\in GL_n(\mathbb F_2)\), let
\[
C_{GL_n(\mathbb F_2)}(T)
=
\{A\in GL_n(\mathbb F_2):AT=TA\}
\]
denote the group centralizer of \(T\). We also write
\[
Z(T)
=
\{L\in \operatorname{End}_{\mathbb F_2}(\mathbb F_{2^n}):LT=TL\}
\]
for its centralizer algebra. These two objects are related by
\[
C_{GL_n(\mathbb F_2)}(T)=Z(T)^\times.
\]

The connection with the conjugator sets from the previous section is
straightforward but important. If \(C(T,U)\neq\varnothing\), then
\(C(T,U)\) is a coset of \(C_{GL_n(\mathbb F_2)}(T)\). In particular,
\[
|C(T,U)|=|C_{GL_n(\mathbb F_2)}(T)|.
\]
Thus, controlling the sizes of the conjugator sets reduces to
understanding the centralizers of the semilinear maps \(T_{a,r}\).

The upper bound for \(p_n\) will follow by summing the sizes of the
conjugator sets appearing in Proposition~\ref{prop:exceptional-union}.
Since a nonempty conjugator set has the same cardinality as the
corresponding group centralizer, it remains to control
\(C_{GL_n(\mathbb F_2)}(T_{a,r})\).

We do this through the centralizer algebra \(Z(T_{a,r})\). For the upper bound on \(p_n\), it is enough to obtain a uniform bound
on the dimension of \(Z(T_{a,r})\). The exact dimension formula is more
refined, we state it below and provide its proof in Appendix~A. 

For \(0\le r<n\), set
\[
d_r=\gcd(n,r),
\]
and let
\[
z_{a,r}
=
N_{\mathbb F_{2^n}/\mathbb F_{2^{d_r}}}(a)
\in \mathbb F_{2^{d_r}}^\times
\]
be the relative norm of \(a\) from \(\mathbb F_{2^n}\) to
\(\mathbb F_{2^{d_r}}\). Finally, define
\[
e_{a,r}
=
[\mathbb F_2(z_{a,r}):\mathbb F_2].
\]

\begin{proposition}
\label{prop:centralizer-dimension}
For \(a\in\mathbb F_{2^n}^{\times}\) and \(0\le r<n\),
\[
\dim_{\mathbb F_2} Z(T_{a,r})
=
\frac{n\,d_r}{e_{a,r}}.
\]
\end{proposition}
\begin{proof}
See Appendix~A. 
\end{proof}

The exact formula immediately yields the uniform estimate needed for
the upper bound on \(p_n\).

\begin{lemma}
\label{lem:uniform-centralizer}
For every \((a,r)\neq(1,0)\),
\[
\dim_{\mathbb F_2} Z(T_{a,r})
\le
\frac{n^2}{q(n)}.
\]
\end{lemma}
\begin{proof}
If \(r=0\), then \(a\neq1\). Hence
\(e_{a,0}=[\mathbb F_2(a):\mathbb F_2]>1\).
Since \(e_{a,0}\mid n\), we have
\(e_{a,0}\ge q(n)\), and therefore
\[
\dim_{\mathbb F_2} Z(T_{a,0})
=
\frac{n^2}{e_{a,0}}
\le
\frac{n^2}{q(n)}.
\]

If \(1\le r<n\), then \(d_r=\gcd(n,r)\) is a proper divisor of \(n\),
so
\[
d_r\le \frac{n}{q(n)}.
\]
Since \(e_{a,r}\ge1\),
\[
\dim_{\mathbb F_2} Z(T_{a,r})
=
\frac{n\,d_r}{e_{a,r}}
\le
n d_r
\le
\frac{n^2}{q(n)}.
\]  
\end{proof}

\begin{remark}
Since
\[
C_{GL_n(\mathbb F_2)}(T_{a,r})=Z(T_{a,r})^\times,
\]
we obtain
\[
|C_{GL_n(\mathbb F_2)}(T_{a,r})|
\le
|Z(T_{a,r})|
=
2^{\dim_{\mathbb F_2} Z(T_{a,r})}
\le
2^{n^2/q(n)}.
\]
\end{remark}

The preceding lemma shows that every nontrivial semilinear candidate has
centralizer dimension at most \(n^2/q(n)\). To obtain a matching lower
bound for \(p_n\), we now show that this scale is realized by a natural
family arising from the Frobenius symmetry of inversion.

Let
\[
m=\frac{n}{q(n)},
\qquad
\tau=\sigma^m.
\]
Since \(\tau=\sigma^m\) is an \(\mathbb F_2\)-linear automorphism, we
may view \(\tau\) as an element of \(GL_n(\mathbb F_2)\). Moreover, \(\tau\) has order \(q(n)\), and its fixed field
\(\mathbb F_{2^m}\) has maximal size among the proper subfields of
\(\mathbb F_{2^n}\).
\begin{lemma}
\label{lem:frobenius-exceptional}
If
\[
A\in C_{GL_n(\mathbb F_2)}(\tau),
\]
then \(A\in\mathcal E_n\). Hence
\[
C_{GL_n(\mathbb F_2)}(\tau)\subseteq \mathcal E_n.
\]
\end{lemma}
\begin{proof}
Since \(\tau\) commutes with both \(A\) and \(J\),
\[
\tau\circ A\circ J
=
A\circ \tau\circ J
=
A\circ J\circ \tau.
\]
Thus
\[
\tau\in\operatorname{Stab}(\Phi_A).
\]
Because \(\tau\neq I\), we have \(A\in\mathcal E_n\).
\end{proof}

It remains to estimate the size of
\(C_{GL_n(\mathbb F_2)}(\tau)\).
The key point is that the action of \(\tau\) admits a natural module
description over the group algebra of the cyclic group generated by
\(\tau\). The following structural description provides a convenient
way to estimate the lower bound. Since its proof is relatively long and
technical, and is not central to the main argument, we provide it in
Appendix~B.

Let \(\langle\tau\rangle\) denote the cyclic group generated by
\(\tau\). Since \(\tau\) has order \(q(n)\), this group has order
\(q(n)\). Let
\[
R=\mathbb F_2[\langle\tau\rangle]
\]
be the group algebra of \(\langle\tau\rangle\) over \(\mathbb F_2\).

We use the \(R\)-module structure induced by the action of
\(\langle\tau\rangle\) to estimate
\(C_{GL_n(\mathbb F_2)}(\tau)\).

\begin{proposition}
\label{prop:tau-centralizer}
Under the natural \(R\)-module structure on \(\mathbb F_{2^n}\)
induced by the action of \(\langle\tau\rangle\),
\[
\mathbb F_{2^n}\cong R^m,
\qquad
m=\frac{n}{q(n)}.
\]
Consequently,
\[
Z(\tau)\cong M_m(R),
\]
where \(M_m(R)\) denotes the ring of \(m\times m\) matrices with
entries in \(R\). Moreover,
\[
C_{GL_n(\mathbb F_2)}(\tau)\cong GL_m(R).
\]
\end{proposition}
\begin{proof}
See Appendix~B.
\end{proof}

The preceding proposition reduces the lower-bound problem to comparing
the number of invertible matrices over \(R\) with the total number of
matrices over \(R\).

Let
\[
\kappa=\prod_{j=1}^{\infty}(1-2^{-j}).
\]
The following lemma gives the required lower bound on the proportion of
invertible matrices in \(M_m(R)\).

\begin{lemma}
\label{lem:invertible-proportion}
With \(R=\mathbb F_2[\langle\tau\rangle]\),
\[
|C_{GL_n(\mathbb F_2)}(\tau)|
=
|GL_m(R)|
\ge
\kappa^{q(n)}|M_m(R)|.
\]
\end{lemma}

\begin{proof}
Since
\[
R=\mathbb F_2[\langle\tau\rangle]
\cong
\mathbb F_2[x]/(x^{q(n)}-1),
\]
we distinguish the cases \(q(n)\) odd and \(q(n)=2\).

Suppose first that \(q(n)\) is odd. Then \(x^{q(n)}-1\) is square-free
over \(\mathbb F_2\). Hence
\[
R\cong \prod_{i=1}^s \mathbb F_{2^{f_i}}
\]
for some positive integers \(f_1,\ldots,f_s\), with \(s\le q(n)\).
It follows that
\[
M_m(R)\cong
\prod_{i=1}^s M_m(\mathbb F_{2^{f_i}})
\]
and
\[
GL_m(R)\cong
\prod_{i=1}^s GL_m(\mathbb F_{2^{f_i}}).
\]
Therefore
\[
\frac{|GL_m(R)|}{|M_m(R)|}
=
\prod_{i=1}^s
\prod_{j=1}^m
\left(1-2^{-f_i j}\right).
\]
Since \(f_i\ge1\),
\[
\prod_{j=1}^m
\left(1-2^{-f_i j}\right)
\ge
\prod_{j=1}^{\infty}
\left(1-2^{-j}\right)
=
\kappa,
\]
and hence
\[
\frac{|GL_m(R)|}{|M_m(R)|}
\ge
\kappa^s
\ge
\kappa^{q(n)}.
\]

Now suppose that \(q(n)=2\). Then
\[
R\cong
\mathbb F_2[x]/((x+1)^2).
\]
Let \(\mathfrak m\) be the ideal generated by the image of \(x+1\).
Then \(R/\mathfrak m\cong\mathbb F_2\). Since \(R\) is a local ring,
an element of \(R\) is a unit if and only if its image in
\(R/\mathfrak m\) is nonzero. Therefore, for \(A\in M_m(R)\),
\[
A\in GL_m(R)
\iff
\det(A)\in R^\times
\iff
\det(\overline A)\neq0,
\]
where \(\overline A\) denotes the image of \(A\) in
\(M_m(R/\mathfrak m)\). Hence \(A\) is invertible over \(R\) if and
only if \(\overline A\) is invertible over \(\mathbb F_2\). Consequently,
\[
\frac{|GL_m(R)|}{|M_m(R)|}
=
\frac{|GL_m(\mathbb F_2)|}{|M_m(\mathbb F_2)|}
=
\prod_{j=1}^m(1-2^{-j})
\ge
\kappa
\ge
\kappa^{q(n)}.
\]
Thus
\[
|GL_m(R)|
\ge
\kappa^{q(n)}|M_m(R)|.
\]
Since
\[
|C_{GL_n(\mathbb F_2)}(\tau)|=|GL_m(R)|,
\]
the result follows.
\end{proof}

\begin{remark}
Since \(R=\mathbb F_2[\langle\tau\rangle]\) has
\[
|R|=2^{q(n)},
\]
we have
\[
|M_m(R)|
=
|R|^{m^2}
=
2^{q(n)m^2}
=
2^{n^2/q(n)}.
\]
Therefore,
\[
|C_{GL_n(\mathbb F_2)}(\tau)|
=
|GL_m(R)|
\ge
\kappa^{q(n)}2^{n^2/q(n)}.
\]
\end{remark}
We have now obtained the two structural estimates needed for the main
probabilistic result. On the upper-bound side, every nontrivial
semilinear candidate satisfies
\[
\dim_{\mathbb F_2} Z(T_{a,r})
\le
\frac{n^2}{q(n)},
\]
and therefore
\[
|C_{GL_n(\mathbb F_2)}(T_{a,r})|
\le
2^{n^2/q(n)}.
\]

On the lower-bound side, the Frobenius automorphism
\[
\tau=\sigma^{n/q(n)}
\]
produces a natural family of exceptional outer linear transformations.
Since
\[
|M_m(R)|
=
2^{n^2/q(n)},
\]
we have
\[
|C_{GL_n(\mathbb F_2)}(\tau)|
\ge
\kappa^{q(n)}2^{n^2/q(n)}.
\]

Thus the upper and lower estimates are governed by the same
centralizer scale $
2^{n^2/q(n)}.
$
We now combine these bounds to prove the main theorem.

\section{Proof of the Main Theorem}
We now combine the structural estimates from the previous section to
obtain matching upper and lower bounds for \(p_n\).
\begin{maintheorem}
\label{thm:main}
For \(n\ge4\),
\[
\kappa^{q(n)}
2^{-n^2(1-1/q(n))}
\le
p_n
\le
\frac{n}{\kappa}
2^{-n^2(1-1/q(n))+n},
\]
where
\[
\kappa=\prod_{j=1}^{\infty}(1-2^{-j}).
\]
\end{maintheorem}
\begin{proof}
We first prove the upper bound. By
Proposition~\ref{prop:exceptional-union},
\[
\mathcal E_n
=
\bigcup_{(a,r)\neq(1,0)}
C(T_{a,r},U_{a,r}).
\]
Hence
\[
|\mathcal E_n|
\le
\sum_{(a,r)\neq(1,0)}
|C(T_{a,r},U_{a,r})|.
\]
Whenever \(C(T_{a,r},U_{a,r})\) is nonempty, it has the same
cardinality as \(C_{GL_n(\mathbb F_2)}(T_{a,r})\). By
Lemma~\ref{lem:uniform-centralizer},
\[
|C_{GL_n(\mathbb F_2)}(T_{a,r})|
\le
2^{n^2/q(n)}.
\]
There are
\[
n(2^n-1)-1<n2^n
\]
nontrivial pairs \((a,r)\). Therefore
\[
|\mathcal E_n|
<
n2^n\,2^{n^2/q(n)}.
\]

Since
\[
|GL_n(\mathbb F_2)|
=
2^{n^2}
\prod_{j=1}^n(1-2^{-j})
\ge
\kappa\,2^{n^2},
\]
we obtain
\[
p_n
=
\frac{|\mathcal E_n|}{|GL_n(\mathbb F_2)|}
\le
\frac{n}{\kappa}
2^{-n^2(1-1/q(n))+n}.
\]

For the lower bound, let
\[
m=\frac{n}{q(n)},
\qquad
\tau=\sigma^m.
\]
By Lemma~\ref{lem:frobenius-exceptional},
\[
C_{GL_n(\mathbb F_2)}(\tau)
\subseteq
\mathcal E_n.
\]
Hence
\[
p_n
\ge
\frac{|C_{GL_n(\mathbb F_2)}(\tau)|}
{|GL_n(\mathbb F_2)|}.
\]
By Lemma~\ref{lem:invertible-proportion},
\[
|C_{GL_n(\mathbb F_2)}(\tau)|
\ge
\kappa^{q(n)}
2^{n^2/q(n)}.
\]
Since
\[
|GL_n(\mathbb F_2)|
\le
2^{n^2},
\]
it follows that
\[
p_n
\ge
\kappa^{q(n)}
2^{-n^2(1-1/q(n))}.
\]
This proves the theorem.
\end{proof}

The two-sided estimate in Main Theorem determines the leading
quadratic exponent of the exceptional probability. Taking logarithms
gives the following asymptotic form.
\begin{corollary}
For \(n\ge4\),
\[
-\log_2 p_n
=
n^2\left(1-\frac1{q(n)}\right)+O(n).
\]
\end{corollary}
\begin{proof}
Taking \(-\log_2\) of the two-sided bound in the Main Theorem reverses the inequalities and gives
\[
n^2\left(1-\frac1{q(n)}\right)
-n-\log_2\frac{n}{\kappa}
\le
-\log_2 p_n
\]
and
\[
-\log_2 p_n
\le
n^2\left(1-\frac1{q(n)}\right)
+q(n)\log_2\frac1\kappa.
\]
Since \(q(n)\le n\), both error terms are \(O(n)\). Therefore
\[
-\log_2 p_n
=
n^2\left(1-\frac1{q(n)}\right)+O(n).
\]
\end{proof}
\begin{remark}
The quadratic exponent is determined by the smallest
prime divisor of \(n\). If \(n\) is even, then \(q(n)=2\), and
\[
-\log_2 p_n
=
\frac{n^2}{2}+O(n).
\]
If \(n\) is prime, then \(q(n)=n\), and
\[
-\log_2 p_n
=
n^2-n+O(n)
=
n^2+O(n).
\]
\end{remark}
Thus the probability that a random outer linear transformation preserves
a nontrivial linear self-symmetry of inversion is exponentially small
in \(n^2\), with a leading exponent determined by the arithmetic of
the dimension through \(q(n)\). The matching lower bound shows that
this exponent is not an artifact of the upper-bound argument. It is already realized by the Frobenius symmetry associated with a
proper subfield of maximal size.

Table~\ref{tab:small-n} illustrates the theorem for
\(4\le n\le 8\). For \(n=4\), exhaustive enumeration gives the exact
value \(p_4=5/168\); for \(n=5,6,7,8\), we report the bounds from
Main Theorem.

\begin{table}[t]
\centering
\begin{tabular}{c|c|c|c}
$n$ & $q(n)$ & $n^2(1-1/q(n))$ & $p_n$ or theorem bounds \\
\hline
4 & 2 & 8
  & $p_4=\dfrac{5}{168}\approx 2.98\times10^{-2}$ \\[2mm]
5 & 5 & 20
  & $1.92\times10^{-9}\le p_5\le 5.28\times10^{-4}$ \\
6 & 2 & 18
  & $3.18\times10^{-7}\le p_6\le 5.07\times10^{-3}$ \\
7 & 7 & 42
  & $3.81\times10^{-17}\le p_7\le 7.05\times10^{-10}$ \\
8 & 2 & 32
  & $1.94\times10^{-11}\le p_8\le 1.65\times10^{-6}$
\end{tabular}
\caption{Small-dimensional values and bounds for the exceptional
probability \(p_n\).}
\label{tab:small-n}
\end{table}

\section{Cryptographic Applications}
\label{sec:applications}

The preceding results describe how frequently a linear self-symmetry
of inversion survives after an outer linear transformation. In this
section, we turn to concrete cryptographic applications. We first
derive a simple rigidity criterion based on the coefficient
relations arising from the stabilizer condition. The criterion depends on only two
coefficients of the linearized polynomial of the outer linear map.
We then apply it to the AES S-box. After that, we discuss how such
self-symmetries can interact with diffusion layers in iterative
constructions, with ARIA and Vision as examples.

\paragraph{A General Rigidity Criterion.}
Let \(n\ge4\), and write the linearized polynomial of an invertible
\(\mathbb F_2\)-linear map \(A\) as
\[
A(x)=\sum_{i=0}^{n-1}\lambda_i x^{2^i}.
\]

\begin{proposition}
\label{prop:rigidity-criterion}
Suppose that \(\lambda_0\lambda_1\ne0\), and set
\[
\rho=\frac{\lambda_1^2}{\lambda_0^3}.
\]
If \(\rho\) does not belong to any proper subfield of
\(\mathbb F_{2^n}\), then
\[
\operatorname{Stab}(A\circ J)=\{I\}.
\]
\end{proposition}

\begin{proof}
Let \(T\in\operatorname{Stab}(A\circ J)\). By
Theorem~\ref{thm:yuan}, there exist
\(a\in\mathbb F_{2^n}^{\times}\) and \(0\le r<n\) such that
\[
T=M_a\sigma^r
\]
and
\[
A^{-1}M_a\sigma^r A=M_{a^{-1}}\sigma^r.
\]
Equivalently,
\[
M_a\sigma^r A=A M_{a^{-1}}\sigma^r.
\]
Expanding both sides gives
\[
\sum_{i=0}^{n-1}a\lambda_i^{2^r}x^{2^{i+r}}
=
\sum_{i=0}^{n-1}\lambda_i a^{-2^i}x^{2^{i+r}}.
\]
Since \(i\mapsto i+r\pmod n\) permutes the Frobenius exponents,
comparison of coefficients gives
\[
a^{2^i+1}=\lambda_i^{\,1-2^r}
\]
for every \(i\) with \(\lambda_i\ne0\).
In particular,
\[
a^2=\lambda_0^{\,1-2^r},
\qquad
a^3=\lambda_1^{\,1-2^r}.
\]
Therefore
\[
\rho^{\,1-2^r}=1,
\]
and hence
\[
\rho^{2^r}=\rho.
\]

If \(r\ne0\), then
\[
\rho\in\mathbb F_{2^{\gcd(n,r)}},
\]
which is a proper subfield of \(\mathbb F_{2^n}\), contradicting the
hypothesis. Thus \(r=0\). The relation for \(i=0\) then gives
\(a^2=1\), so \(a=1\). Hence \(T=I\), and therefore
\[
\operatorname{Stab}(A\circ J)=\{I\}.
\]
\end{proof}
\paragraph{The AES S-box.}

We now apply Proposition~\ref{prop:rigidity-criterion} to the AES
S-box~\cite{fips197}. In the standard field representation
\[
\mathbb F_{2^8}
=
\mathbb F_2[t]/(t^8+t^4+t^3+t+1),
\]
the linear part of the AES affine transformation has linearized
polynomial
\[
A_{\rm AES}(x)
=
\sum_{i=0}^7 \lambda_i x^{2^i},
\]
with
\[
(\lambda_0,\ldots,\lambda_7)
=
(05,09,F9,25,F4,01,B5,8F).
\]

For the first two coefficients,
\[
\rho
=
\frac{\lambda_1^2}{\lambda_0^3}
=
E7.
\]
Every proper subfield of \(\mathbb F_{2^8}\) is contained in
\(\mathbb F_{2^4}\). Since
\[
E7^{15}=35\ne1,
\]
we have \(E7\notin\mathbb F_{2^4}\). Hence \(E7\) belongs to no proper
subfield of \(\mathbb F_{2^8}\).

Proposition~\ref{prop:rigidity-criterion} therefore gives
\[
\operatorname{Stab}(A_{\rm AES}\circ J_8)=\{I\}.
\]

The full AES S-box is
\[
S_{\rm AES}(x)
=
A_{\rm AES}(J_8(x))+c,
\]
where \(c=63\) in hexadecimal notation.

\begin{corollary}\label{AESSbox}
The full AES S-box has no nontrivial linear self-symmetry:
\[
\operatorname{Stab}(S_{\rm AES})=\{I\}.
\]
\end{corollary}

\begin{proof}
Let \(T\in\operatorname{Stab}(S_{\rm AES})\). Evaluating
\(T\circ S_{\rm AES}=S_{\rm AES}\circ T\) at \(0\) gives
\(T(c)=c\). The constant terms therefore cancel, and
\[
T\in\operatorname{Stab}(A_{\rm AES}\circ J_8).
\]
Hence \(T=I\).
\end{proof}

Thus the AES affine layer satisfies the two-coefficient criterion, and
the resulting S-box has no nontrivial linear self-symmetry. We next
examine why such S-box-level self-symmetries can matter at the level of
an iterative cipher.

\paragraph{ARIA and Binary Diffusion.}
ARIA~\cite{lee2010aria} uses two S-boxes and their inverses. The first,
\(S_1\), is the AES S-box. The second is
\[
S_2(x)=B\bigl(x^{247}\bigr)+c,
\]
for an affine layer with linear part \(B\). Since
\[
x^{247}=x^{-8}=\sigma^3(J(x)),
\]
the nonconstant part of \(S_2\) belongs to the family considered here,
with outer linear layer \(B\circ\sigma^3\).

The diffusion layer of ARIA is a \(16\times16\) binary matrix acting
on bytes. Hence any \(\mathbb F_2\)-linear map applied identically to
each byte commutes with this diffusion layer. Therefore, a byte-wise
linear self-symmetry that survives the S-box layer is not destroyed by
the diffusion layer and may continue through the round structure if it
is also compatible with the remaining round operations.

For ARIA, S-box rigidity is therefore the natural condition that
prevents such byte-wise linear self-symmetries from propagating through
the diffusion layer. Corollary~\ref{AESSbox} already gives this rigidity
for \(S_1\).

\paragraph{A Large-Field Example: Vision.}
The Vision cipher~\cite{aly2020marvellous} uses S-boxes obtained by composing
inversion over a binary field with an \(\mathbb F_2\)-linearized affine
permutation. Writing the linear part as
\[
A(x)=b_1x+b_2x^2+b_3x^4,
\]
the transformed inverse \(A\circ J\) belongs directly to the family
studied in this paper.

This example is different from the byte-oriented setting of AES and
ARIA. The field dimension is much larger, while the outer linear map
has a sparse linearized representation. In this setting,
Proposition~\ref{prop:rigidity-criterion} gives a particularly simple rigidity
test. \(A\circ J\) is rigid whenever
\[
\frac{b_2^2}{b_1^3}
\]
lies in no proper subfield of \(\mathbb F_{2^n}\), provided that $b_1b_2 \neq 0$.

Thus the criterion is not limited to small S-boxes. It also applies
naturally to inverse-based nonlinear layers over larger binary fields,
where sparsity of the outer linearized polynomial makes the rigidity
condition especially easy to check. For such large-field constructions,
rigidity is a desirable structural property. It rules out nontrivial
linear self-symmetries inherited from finite-field inversion. This does
not by itself imply cryptographic security, but it removes one possible
source of algebraic structure that could otherwise persist through an
iterated construction.

\paragraph{Related Inverse-Based Constructions.}
The framework also applies directly to inverse-based S-boxes with
linear transformations on both sides of inversion. Suppose
\[
S=B\circ J\circ A,
\qquad A,B\in GL_n(\mathbb F_2).
\]
Then
\[
A\circ S\circ A^{-1}=(AB)\circ J,
\]
and therefore
\[
\operatorname{Stab}(S)
=
A^{-1}\operatorname{Stab}((AB)\circ J)A.
\]
Thus \(S\) is rigid if and only if \((AB)\circ J\) is rigid. In
particular, two-sided linear transformations around inversion reduce
directly to the outer-linear setting studied in this paper.

Additive constants after inversion can be handled as in
Corollary~\ref{AESSbox}. If a constant is added before inversion,
as in Camellia~\cite{rfc3713} and SM4~\cite{SM4ISO}, one must
conjugate by an affine map. A linear self-symmetry of the original
S-box is then transformed into an affine self-symmetry of a map of
the form
\[
(AB)\circ J+c.
\]
Thus the linear part is covered by the present theory, while the
additional affine condition can be checked separately.

\section{Conclusion}

We studied the linear self-symmetries of finite-field inversion after
composition with an outer invertible linear transformation. We showed
that the exceptional probability satisfies
\[
-\log_2 p_n
=
n^2\left(1-\frac1{q(n)}\right)+O(n),
\]
so the leading quadratic exponent is determined by the smallest prime
divisor \(q(n)\) of the extension degree.

The cryptographic examples illustrate how this symmetry question
appears in practice. The AES S-box satisfies a simple rigidity
criterion, ARIA shows how byte-wise symmetries can pass through binary
diffusion layers, and Vision extends the setting to larger binary
fields with sparse linearized outer maps. The same framework also extends directly to two-sided linear
transformations around inversion, covering the linear parts of related
constructions such as Camellia and SM4.

\bibliographystyle{plain}
\bibliography{reference}

@article{yuan2019inverse,
  author  = {Yuan, Feng and Jiang, Ji-Jun and Yang, Yang and Ou, Hai-Wen and Wang, Min-Juan},
  title   = {Enumeration of Cryptographic Functions Affine Equivalent to the Inverse Function Over {$F_{p^n}$}},
  journal = {Chinese Journal of Computers},
  volume  = {42},
  number  = {5},
  pages   = {1126--1136},
  year    = {2019},
  doi     = {10.11897/SP.J.1016.2019.01126}
}

@techreport{fips197,
  author      = {{National Institute of Standards and Technology}},
  title       = {Advanced Encryption Standard (AES)},
  institution = {National Institute of Standards and Technology},
  type        = {Federal Information Processing Standards Publication},
  number      = {197-upd1},
  year        = {2023},
  doi         = {10.6028/NIST.FIPS.197-upd1}
}

@inproceedings{biryukov2003toolbox,
  author    = {Biryukov, Alex and De Canni{\`e}re, Christophe
               and Braeken, An and Preneel, Bart},
  title     = {A Toolbox for Cryptanalysis:
               Linear and Affine Equivalence Algorithms},
  booktitle = {Advances in Cryptology -- EUROCRYPT 2003},
  series    = {Lecture Notes in Computer Science},
  volume    = {2656},
  pages     = {33--50},
  publisher = {Springer},
  year      = {2003},
  doi       = {10.1007/3-540-39200-9_3}
}

@inproceedings{ranea2021self,
  author    = {Ranea, Adrian and Preneel, Bart},
  title     = {On Self-Equivalence Encodings in White-Box Implementations},
  booktitle = {Selected Areas in Cryptography -- SAC 2020},
  series    = {Lecture Notes in Computer Science},
  volume    = {12804},
  pages     = {639--669},
  publisher = {Springer},
  year      = {2021},
  doi       = {10.1007/978-3-030-81652-0_25}
}

@article{rijmen2008rotation,
  author  = {Rijmen, Vincent and Barreto, Paulo S. L. M.
             and Gazzoni Filho, D{\'e}cio L.},
  title   = {Rotation Symmetry in Algebraically Generated
             Cryptographic Substitution Tables},
  journal = {Information Processing Letters},
  volume  = {106},
  number  = {6},
  pages   = {246--250},
  year    = {2008},
  doi     = {10.1016/j.ipl.2007.09.012}
}

@article{rosenthal2003polynomial,
  author  = {Rosenthal, Joachim},
  title   = {A Polynomial Description of the {Rijndael}
             Advanced Encryption Standard},
  journal = {Journal of Algebra and Its Applications},
  volume  = {2},
  number  = {2},
  pages   = {223--236},
  year    = {2003},
  doi     = {10.1142/S0219498803000532}
}

@article{liu2005rijndael,
  author  = {Liu, Jingmei and Wei, Baodian and Cheng, Xiangguo
             and Wang, Xinmei},
  title   = {Cryptanalysis of {Rijndael} {S-box} and Improvement},
  journal = {Applied Mathematics and Computation},
  volume  = {170},
  number  = {2},
  pages   = {958--975},
  year    = {2005},
  doi     = {10.1016/j.amc.2004.12.043}
}

@article{kolsch2021ccz,
  author  = {Lukas K{\"o}lsch},
  title   = {On {CCZ}-Equivalence of the Inverse Function},
  journal = {IEEE Transactions on Information Theory},
  volume  = {67},
  number  = {7},
  pages   = {4856--4862},
  year    = {2021},
  doi     = {10.1109/TIT.2021.3065068}
}

@misc{ishizuka2026stabilizers,
  author = {Keita Ishizuka},
  title  = {Almost All Vectorial Functions Have Trivial Extended-Affine Stabilizers},
  year   = {2026},
  note   = {arXiv:2602.06668. \url{https://arxiv.org/abs/2602.06668}}
}

@article{kavut2012rotation,
  author  = {Kavut, Sel{\c{c}}uk},
  title   = {Results on Rotation-Symmetric {S}-Boxes},
  journal = {Information Sciences},
  volume  = {201},
  pages   = {93--113},
  year    = {2012},
  doi     = {10.1016/j.ins.2012.02.030}
}

@misc{omland2022rotation,
  author = {Omland, Tron and St\u{a}nic\u{a}, Pantelimon},
  title  = {Permutation Rotation-Symmetric {S}-Boxes, Liftings and Affine Equivalence},
  year   = {2022},
  note   = {arXiv:2203.00778}
}

@inproceedings{barkan2002rijndael,
  author    = {Barkan, Elad and Biham, Eli},
  title     = {In How Many Ways Can You Write {Rijndael}?},
  booktitle = {Advances in Cryptology -- ASIACRYPT 2002},
  series    = {Lecture Notes in Computer Science},
  volume    = {2501},
  pages     = {160--175},
  publisher = {Springer},
  year      = {2002},
  doi       = {10.1007/3-540-36178-2_10}
}

@article{carlet1998codes,
  author  = {Carlet, Claude and Charpin, Pascale and Zinoviev, Victor},
  title   = {Codes, Bent Functions and Permutations Suitable for {DES}-Like Cryptosystems},
  journal = {Designs, Codes and Cryptography},
  volume  = {15},
  number  = {2},
  pages   = {125--156},
  year    = {1998},
  doi     = {10.1023/A:1008344232130}
}

@article{canteaut2019ccz,
  author  = {Canteaut, Anne and Perrin, L{\'e}o},
  title   = {On {CCZ}-Equivalence, Extended-Affine Equivalence, and Function Twisting},
  journal = {Finite Fields and Their Applications},
  volume  = {56},
  pages   = {209--246},
  year    = {2019},
  doi     = {10.1016/j.ffa.2018.11.008}
}

@article{chen2021affine,
  author  = {Chen, Xi and Qu, Longjiang and Fu, Shaojing and Li, Chao},
  title   = {The Number of Affine Equivalent Classes and Extended Affine Equivalent Classes of Vectorial Boolean Functions},
  journal = {Discrete Applied Mathematics},
  volume  = {289},
  pages   = {477--491},
  year    = {2021},
  doi     = {10.1016/j.dam.2020.10.002}
}

@techreport{lee2010aria,
  author      = {Lee, Jaechul and Lee, Jung Hee and Kim, Jong-Sung
                 and Kwon, Daesung and Kim, Changhoon},
  title       = {A Description of the ARIA Encryption Algorithm},
  institution = {Internet Engineering Task Force},
  number      = {RFC 5794},
  year        = {2010}
}

@article{aly2020marvellous,
  author  = {Aly, Abdelrahaman and Ashur, Tomer and Ben-Sasson, Eli
             and Dhooghe, Siemen and Szepieniec, Alan},
  title   = {Design of Symmetric-Key Primitives for Advanced
             Cryptographic Protocols},
  journal = {IACR Transactions on Symmetric Cryptology},
  year    = {2020},
  volume  = {2020},
  number  = {3},
  pages   = {1--45},
  doi     = {10.13154/tosc.v2020.i3.1-45}
}

@techreport{SM4ISO,
  author      = {{International Organization for Standardization}},
  title       = {Information Technology---Security Techniques---Encryption Algorithms---Part 3: Block Ciphers---Amendment 1: SM4},
  institution = {ISO/IEC},
  number      = {18033-3:2010/Amd 1:2021},
  year        = {2021}
}

@techreport{rfc3713,
  author      = {Matsui, M. and Nakajima, J. and Moriai, S.},
  title       = {A Description of the Camellia Encryption Algorithm},
  institution = {Internet Engineering Task Force},
  type        = {RFC},
  number      = {3713},
  year        = {2004}
}

@inproceedings{leander2011invariant,
  author    = {Gregor Leander and Mohamed Ahmed Abdelraheem and
               Hoda AlKhzaimi and Erik Zenner},
  title     = {A Cryptanalysis of {PRINTcipher}: The Invariant Subspace Attack},
  booktitle = {Advances in Cryptology -- CRYPTO 2011},
  series    = {Lecture Notes in Computer Science},
  volume    = {6841},
  publisher = {Springer},
  year      = {2011}
}

\appendix 
\section{Proof of Proposition~\ref{prop:centralizer-dimension}}

To prove Proposition~\ref{prop:centralizer-dimension}, we first isolate
the finite-field equation that governs each coefficient of a commuting
linearized polynomial.

\begin{lemma}
\label{lem:hilbert90-space}
Let \(0\le r<n\), let \(d_r=\gcd(n,r)\), and let
\(\eta\in\mathbb F_{2^n}^{\times}\). Set
\[
V_\eta
=
\{u\in\mathbb F_{2^n}:u^{2^r}=\eta u\}.
\]
Then \(V_\eta\) contains a nonzero element if and only if
\[
N_{\mathbb F_{2^n}/\mathbb F_{2^{d_r}}}(\eta)=1.
\]
When this condition holds, \(V_\eta\) is one-dimensional over
\(\mathbb F_{2^{d_r}}\), and hence
\[
\dim_{\mathbb F_2}V_\eta=d_r.
\]
Otherwise, \(V_\eta=\{0\}\).
\end{lemma}

\begin{proof}
If \(r=0\), then \(d_r=n\), and the equation becomes
\[
u=\eta u.
\]
Thus \(V_\eta=\mathbb F_{2^n}\) when \(\eta=1\), and
\(V_\eta=\{0\}\) otherwise, which agrees with the stated norm condition.

Now assume \(1\le r<n\). Since \(d_r=\gcd(n,r)\), the fixed field of
\(\sigma^r\) is \(\mathbb F_{2^{d_r}}\). Moreover, \(\sigma^r\) has
order \(n/d_r\) and generates
\[
\operatorname{Gal}
(\mathbb F_{2^n}/\mathbb F_{2^{d_r}}).
\]
Hence \(V_\eta\) is an \(\mathbb F_{2^{d_r}}\)-subspace of
\(\mathbb F_{2^n}\). If \(u_0,u\in V_\eta\) are nonzero, then
\[
\left(\frac{u}{u_0}\right)^{2^r}
=
\frac{u}{u_0},
\]
so \(u/u_0\in\mathbb F_{2^{d_r}}\). Therefore, whenever \(V_\eta\)
contains a nonzero element, it is one-dimensional over
\(\mathbb F_{2^{d_r}}\).

It remains to determine when a nonzero solution exists. For \(u\ne0\),
\[
u^{2^r}=\eta u
\quad\Longleftrightarrow\quad
\eta=u^{2^r-1}.
\]
By the finite-field form of Hilbert's Theorem~90, this occurs if and
only if
\[
N_{\mathbb F_{2^n}/\mathbb F_{2^{d_r}}}(\eta)=1.
\]
Therefore, when the norm condition holds,
\[
\dim_{\mathbb F_2}V_\eta=d_r,
\]
and otherwise \(V_\eta=\{0\}\).
\end{proof}
We now apply Lemma~\ref{lem:hilbert90-space} coefficientwise to determine
the centralizer algebra \(Z(T_{a,r})\).

\begin{proof}[Proof of Proposition~\ref{prop:centralizer-dimension}]
Every \(\mathbb F_2\)-linear map
\(L:\mathbb F_{2^n}\to\mathbb F_{2^n}\) has a unique linearized
polynomial representation
\[
L(x)=\sum_{i=0}^{n-1}\ell_i x^{2^i},
\qquad
\ell_i\in\mathbb F_{2^n}.
\]
Since
\[
T_{a,r}(x)=a x^{2^r},
\]
we obtain
\[
L\circ T_{a,r}(x)
=
\sum_{i=0}^{n-1}
\ell_i a^{2^i} x^{2^{i+r}},
\]
whereas
\[
T_{a,r}\circ L(x)
=
\sum_{i=0}^{n-1}
a\ell_i^{2^r} x^{2^{i+r}}.
\]
Because addition by \(r\) permutes the Frobenius exponents modulo \(n\),
uniqueness of the linearized polynomial representation gives
\[
L\circ T_{a,r}=T_{a,r}\circ L
\quad\Longleftrightarrow\quad
\ell_i^{2^r}
=
a^{2^i-1}\ell_i
\qquad (0\le i<n).
\]

Thus the centralizer condition reduces to \(n\) independent equations,
one for each coefficient \(\ell_i\).
Applying Lemma~\ref{lem:hilbert90-space} with
\(\eta=a^{2^i-1}\), the \(i\)-th coefficient admits a nonzero solution
if and only if
\[
N_{\mathbb F_{2^n}/\mathbb F_{2^{d_r}}}
\bigl(a^{2^i-1}\bigr)=1.
\]
Using multiplicativity of the norm and the definition of $z_{a,r},$
we obtain
\[
N_{\mathbb F_{2^n}/\mathbb F_{2^{d_r}}}
\bigl(a^{2^i-1}\bigr)
=
N_{\mathbb F_{2^n}/\mathbb F_{2^{d_r}}}(a)^{2^i-1}
=
z_{a,r}^{2^i-1}.
\]
Hence the norm condition is equivalent to
\[
z_{a,r}^{2^i-1}=1,
\]
or equivalently,
\[
z_{a,r}^{2^i}=z_{a,r}.
\]

Now \(e_{a,r}=[\mathbb F_2(z_{a,r}):\mathbb F_2]\), so
\(z_{a,r}^{2^i}=z_{a,r}\) holds exactly when
\(e_{a,r}\mid i\). Moreover, \(z_{a,r}\in\mathbb F_{2^{d_r}}\),
so \(e_{a,r}\mid d_r\), and hence \(e_{a,r}\mid n\). Therefore
exactly \(n/e_{a,r}\) indices \(i\in\{0,\ldots,n-1\}\) satisfy the
norm condition.

For each such index, Lemma~\ref{lem:hilbert90-space} gives a
\(d_r\)-dimensional \(\mathbb F_2\)-solution space for \(\ell_i\).
For all remaining indices, only \(\ell_i=0\) is possible. Since the
coefficient equations are independent,
\[
\dim_{\mathbb F_2} Z(T_{a,r})
=
\frac{n}{e_{a,r}}\,d_r
=
\frac{n\,d_r}{e_{a,r}}.
\]
\end{proof}

\section{Proof of Proposition~\ref{prop:tau-centralizer}}

Recall that
\[
m=\frac{n}{q(n)},
\qquad
\tau=\sigma^m,
\]
so that \(\tau\) has order \(q(n)\) and fixes
\(\mathbb F_{2^m}\). Let
\[
R=\mathbb F_2[\langle\tau\rangle].
\]
The action of \(\langle\tau\rangle\) on \(\mathbb F_{2^n}\) gives
\(\mathbb F_{2^n}\) a natural \(R\)-module structure.

\begin{proof}[Proof of Proposition~\ref{prop:tau-centralizer}]
Consider the Galois extension
$
\mathbb F_{2^n}/\mathbb F_{2^m}.
$
Its Galois group is
\[
\operatorname{Gal}
\bigl(\mathbb F_{2^n}/\mathbb F_{2^m}\bigr)
=
\langle\tau\rangle,
\]
which has order \(q(n)\). By the Normal Basis Theorem,
\(\mathbb F_{2^n}\) is a free module of rank one over
$
\mathbb F_{2^m}[\langle\tau\rangle].
$
Equivalently,
\[
\mathbb F_{2^n}
\cong
\mathbb F_{2^m}[\langle\tau\rangle]
\]
as \(\mathbb F_{2^m}[\langle\tau\rangle]\)-modules.

Since
$
R=\mathbb F_2[\langle\tau\rangle]
$
and \(\mathbb F_{2^m}\) has dimension \(m\) over \(\mathbb F_2\),
we have, as \(R\)-modules,
\[
\mathbb F_{2^m}[\langle\tau\rangle]
\cong
\mathbb F_{2^m}\otimes_{\mathbb F_2}R
\cong
R^m.
\]
Hence
\[
\mathbb F_{2^n}\cong R^m
\]
as \(R\)-modules.

Now let
\[
Z(\tau)
=
\{L\in
\operatorname{End}_{\mathbb F_2}(\mathbb F_{2^n})
:
L\tau=\tau L\}.
\]
Because \(R\) is generated over \(\mathbb F_2\) by the action of
\(\tau\), an \(\mathbb F_2\)-linear map commutes with \(\tau\) if and
only if it is \(R\)-linear. Therefore
\[
Z(\tau)
=
\operatorname{End}_R(\mathbb F_{2^n}).
\]
Using the \(R\)-module isomorphism above,
\[
Z(\tau)
\cong
\operatorname{End}_R(R^m)
\cong
M_m(R).
\]
Here \(R\) is commutative, since \(\langle\tau\rangle\) is cyclic.

Finally,
\[
C_{GL_n(\mathbb F_2)}(\tau)
=
Z(\tau)^\times,
\]
because the group centralizer consists precisely of the invertible
\(\mathbb F_2\)-linear maps in \(Z(\tau)\). Under the isomorphism
\[
Z(\tau)\cong M_m(R),
\]
the invertible elements correspond to the invertible \(m\times m\)
matrices over \(R\). Thus
\[
C_{GL_n(\mathbb F_2)}(\tau)
\cong
GL_m(R).
\]
\end{proof}

\end{document}